\documentclass[reqno,american]{amsart}
\usepackage[T1]{fontenc}
\usepackage[utf8]{inputenc}
\usepackage{babel}
\usepackage{prettyref}
\usepackage{mathrsfs}
\usepackage{mathtools}
\usepackage{amstext}
\usepackage{amsthm}
\usepackage{amssymb}
\usepackage[pdfusetitle,
 bookmarks=true,bookmarksnumbered=false,bookmarksopen=false,
 breaklinks=false,pdfborder={0 0 0},pdfborderstyle={},backref=false,colorlinks=false]
 {hyperref}
\hypersetup{
 colorlinks=true,citecolor=blue,linkcolor=blue,linktocpage=true}

\makeatletter
\numberwithin{equation}{section}
\numberwithin{figure}{section}
\theoremstyle{plain}
\newtheorem{thm}{\protect\theoremname}[section]
\theoremstyle{remark}
\newtheorem{rem}[thm]{\protect\remarkname}
\theoremstyle{plain}
\newtheorem{prop}[thm]{\protect\propositionname}
\theoremstyle{definition}
\newtheorem{defn}[thm]{\protect\definitionname}
\theoremstyle{plain}
\newtheorem{lem}[thm]{\protect\lemmaname}
\theoremstyle{plain}
\newtheorem{cor}[thm]{\protect\corollaryname}

\@ifundefined{date}{}{\date{}}
\usepackage{mathrsfs}

\newrefformat{cor}{Corollary~\ref{#1}}
\newrefformat{subsec}{Section~\ref{#1}}
\newrefformat{lem}{Lemma~\ref{#1}}
\newrefformat{thm}{Theorem~\ref{#1}}
\newrefformat{sec}{Section~\ref{#1}}
\newrefformat{chap}{Chapter~\ref{#1}}
\newrefformat{prop}{Proposition~\ref{#1}}
\newrefformat{exa}{Example~\ref{#1}}
\newrefformat{tab}{Table~\ref{#1}}
\newrefformat{rem}{Remark~\ref{#1}}
\newrefformat{def}{Definition~\ref{#1}}
\newrefformat{fig}{Figure~\ref{#1}}
\newrefformat{claim}{Claim~\ref{#1}}

\AtBeginDocument{
  
}

\makeatother

\providecommand{\corollaryname}{Corollary}
\providecommand{\definitionname}{Definition}
\providecommand{\lemmaname}{Lemma}
\providecommand{\propositionname}{Proposition}
\providecommand{\remarkname}{Remark}
\providecommand{\theoremname}{Theorem}

\begin{document}
\title[]{From Classical to Quantum: Polymorphisms in Non-Commutative Probability}
\author{Palle E.T. Jorgensen}
\address{(Palle E.T. Jorgensen) Department of Mathematics, The University of
Iowa, Iowa City, IA 52242-1419, U.S.A.}
\email{palle-jorgensen@uiowa.edu}
\author{James Tian}
\address{(James F. Tian) Mathematical Reviews, 416 4th Street Ann Arbor, MI
48103-4816, U.S.A.}
\email{jft@ams.org}
\begin{abstract}
We present a parallel between commutative and non-commutative polymorphisms.
Our emphasis is the applications to conditional distributions from
stochastic processes. In the classical case, both the measures and
the positive definite kernels are scalar valued. But the non-commutative
framework (as motivated by quantum theory) dictates a setting where
instead now both the measures (in the form of quantum states), and
the positive definite kernels, are operator valued. The non-commutative
theory entails a systematic study of positive operator valued measures,
abbreviated POVMs. And quantum states (normal states) are indexed
by normalized positive trace-class operators. In the non-commutative
theory, the parallel to the commutative/scalar valued theory helps
us understand entanglement in quantum information. A further implication
of our study of the non-commutative framework will entail an interplay
between the two cases, scalar valued, vs operator valued.
\end{abstract}


\maketitle
\tableofcontents{}

\section{Introduction}

Motivated by questions in quantum information, we suggest here a noncommutative
setting for multi-variable measurement problems. With this we extend
both key notions and results from the framework of classical polymorphism
and Rokhlin disintegration theory to a particular non-commutative
setting. As is in the nature of quantum measurements \cite{MR3956386,MR4771441,MR4764347}
such extensions will by necessity involve choices.

We introduce ``Quantum Polymorphisms'' as an extension of classical
polymorphisms (joint probability distributions with specified marginals)
to the quantum (non-commutative) setting. This quantum version (see
\prettyref{def:4-1} and its variants in \prettyref{sec:5}) considers
positive operator-valued measures (POVMs) acting on a Hilbert space,
as well as operator-valued kernels and their associated Radon-Nikodym
derivatives. This provides a way to model conditional distributions
and correlations, with particular relevance to quantum information
and entanglement.

A key question we address is how classical concepts like polymorphisms
and conditional distributions can be extended to the quantum setting.
Specifically, our framework focuses on understanding quantum correlations,
including those arising from entanglement, through tools such as positive
operator-valued measures (POVMs) and operator-valued kernels. Our
disintegration theorems (Theorems \ref{thm:4-2} and \ref{thm:4-5})
provide a rigorous basis for describing conditional structures in
quantum systems, which are crucial for analyzing the behavior of entangled
states. These results help formalize how quantum measurements reveal
correlations that cannot be explained classically, offering insights
into the role of entanglement in quantum information and measurement
theory.

Our paper has two aims: we (1) present a non-commutative framework
(\prettyref{sec:4}), and (2) offer a noncommutative Rokhlin-type
result that serves as a direct integral disintegration (Theorems \ref{thm:4-2},
\ref{thm:4-5}).

For a direct link between Rokhlin/polymorphism on the one hand, and
direct integrals in Hilbert space on the other, see e.g., \cite{MR4676388}.
While the use of the rigorous Rokhlin factorizations/direct integrals,
corresponding to classes of polymorphisms, already serves as a powerful
tool in mathematics, we note that it further provides a precise (mathematical)
formulation of conditioning and Bayes’ rules in probability, i.e.,
describing the probability of an event, based on prior knowledge of
conditions influencing the event; see \eqref{eq:b4}, and \eqref{eq:b5}.
We show in Sections \ref{sec:3}--\ref{sec:5} of the paper that
this approach in fact also carries over to the non-commutative case,
i.e., to quantum observables; see especially \eqref{eq:c7}, \eqref{eq:d2-1},
\eqref{eq:D3}, \eqref{eq:D7}, and \eqref{eq:D8}. To help motivate
our results, we reference our earlier papers \cite{MR4730761,jorgensen2024hilbert,JT2024canonical,ncrn},
as well as the papers cited there \cite{MR2743416,MR4287855,MR4414884,MR4439542,MR582649,MR2790067,MR3390587,MR3956386,MR4130407}. 

Our paper aims to address foundational questions, and especially to
link some key ideas from commutative and non-commutative theories.
One aspect of this entails the role of conditional measurements in
both classical and quantum probability, including a framework of the
Bayes' rules in the two contests. However, we want to stress that
there is an earlier relevant literature dealing with related issues,
cited inside our paper. But here we want to first call attention to
some recent and relevant papers by Parzygnat, especially \cite{MR4642246,MR4627098,MR4393903,MR1998572}
dealing also with a wider context of operator algebras and modular
theory.
\begin{rem}
Although our emphasis in this paper is theoretical, we do briefly
call attention to an element of motivation which is derived in part
from quantum entanglement in physics, as well as broader foundational
questions. But for this part, it is not our intention to go into details.
It is vast and diverse subject. Indeed, it is a very vast subject
with multiple and diverse approaches. And so do justice to details
would be a whole new paper (or book.) The literature is vast, so to
help readers interested in all of this, we instead offer just a small
selection of cited papers, see e.g., \cite{MR3081872,MR4546402},
and the sources cited there. While the central themes in our present
paper deal with issues of non-commutativity, we feel that they are
interesting in their own right. But part of our paper indeed are motivated
by the mathematics of non-commutativity which in turn underpin foundational
notions of entanglement in quantum theory. In more detail, the study
of quantum measurements of physical properties such as position, momentum,
spin, and polarization performed on entangled particles can, which
in some cases display correlations reflecting quantum entanglement.

Sone general points: Non-commutative probability originates with foundational
questions from quantum theory, and it entails the study of a variety
of systems of noncommutative random variables (i.e., in the form of
operators in Hilbert space). Note that this framework contrasts with
that of classical (commutative) stochastic analysis where the appropriate
choices of random variables (which are functions), random fields forming
commutative algebras.
\end{rem}

\section{Commutative polymorphisms}\label{sec:2}

The notion and study of polymorphisms were pioneered by Anatoly M.
Vershik, see \cite{MR2166671} and the papers cited below. The basic
idea is that while functions transform sets, multi-valued maps and
other generalizations are needed to study different types of dynamics
in measure theory. Polymorphisms, systematically studied by Vershik,
have many applications (see e.g., \cite{MR3958137,MR4126821,MR4131038,MR4218424,MR4368036,MR4295177,MR4388381}).
These include Markov maps, transition operators, two-dimensional distributions
of Markov processes, and joinings of measure spaces. Similar concepts
in other fields include correspondences in algebra, bifibrations in
differential geometry, and Markov maps in probability theory. In physics,
polymorphisms relate to the ``coarse-graining'' approach to dynamics,
where maps send a point to measures.

Recall that a given positive measure $\nu$ on $X\times X$ ($\left(X,\mathscr{B}\right)$
being a measurable space) relative to the product $\sigma$-algebra
$\mathscr{B}_{2}$ is said to be a polymorphism for $\mu_{1}$ and
$\mu_{2}$ when $\mu_{i}$ are the respective marginal measures $\mu_{i}:=\nu\circ\pi_{i}^{-1}$
relative to the projections $\pi_{i}\left(\left(x_{1},x_{2}\right)\right)=x_{i}$,
$i=1,2$. 

Assume $\left(X,\mathscr{B}\right)$ is a standard measurable space.
Let $\mu_{1}\leftarrow\nu\rightarrow\mu_{2}$ be a polymorphism, then
the marginal measures $\mu_{i}$ yield conditional measures $\nu\left(\cdot\mid x_{i}\right)$
defined as follows: 
\begin{align}
\nu\left(dx_{1}dx_{2}\right) & =\mu_{1}\left(dx_{1}\right)\nu\left(dx_{2}\mid x_{1}\right)\nonumber \\
 & =\mu_{2}\left(dx_{2}\right)\nu\left(dx_{1}\mid x_{2}\right),\label{eq:b4}
\end{align}
or equivalently, with $\nu$ specified on product sets $A\times B$,
$A,B\in\mathscr{B}$: 
\begin{align}
\nu\left(A\times B\right) & =\int_{A}\mu_{1}\left(dx_{1}\right)\nu\left(B\mid x_{1}\right)\nonumber \\
 & =\int_{B}\mu_{2}\left(dx_{2}\right)\nu\left(A\mid x_{2}\right).\label{eq:b5}
\end{align}
Formulas \eqref{eq:b4} and \eqref{eq:b5} are closely related to
Bayes' rules in probability and the slice decompositions of measures
on product spaces in ergodic theory. Bayes' rules describe the probability
of events based on prior knowledge of conditions related to those
events, hence conditional distributions.

In this paper, we consider a quantum version of polymorphisms within
a non-commutative setting. Extending results from the classical, commutative
framework to the non-commutative case allows us to explore new quantum
dynamics.

We will avoid certain technical details. However, for experts, we
follow Vershik in assuming our spaces are standard measure spaces.
We recall that Rokhlin's papers on disintegrations cover this issue
in depth, especially in ``separable'' measure spaces. Our formulation
\prettyref{thm:4-2} and the following sections assume certain properties
for the measure spaces, specifically standard measure spaces. This
presentation is meant to be accessible to beginners. For more details,
readers can refer to \cite{MR4368036,MR4388381,MR2166671}, and to
the papers cited therein.

\section{Positive operator valued measures (POVMs)}\label{sec:3}

In the following, we offer a non-commutative framework for polymorphisms.
While it is motivated by questions in quantum information, we believe
that it is also of independent interest. Recall (see above), that
in the classical case, we deal with a framework where both the measures
and the positive definite kernels are assumed scalar valued. However,
the non-commutative framework (as motivated by quantum theory) dictates
a setting where both the measures (in the form of quantum states)
and the positive definite kernels are operator-valued. Furthermore,
the non-commutative theory entails a systematic study of positive
operator-valued measures (abbreviated as POVMs).

Our focus below will include the setting when quantum states are normal
states (in the sense of operator algebras), i.e., they are indexed
by normalized positive trace-class operators (see e.g., \cite{MR442701,MR230136}).
For the non-commutative counterparts, the parallel commutative/scalar-valued
framework of polymorphisms in \prettyref{sec:2} above helps us understand
entanglement as it arises in quantum information. A further implication
of our study of the non-commutative framework is that addressing important
questions will entail an interplay between the two cases: scalar-valued
vs. operator-valued.

Let $\left(S,\mathscr{B}_{S}\right)$ be a measurable space, where
$S$ is a set, and $\mathscr{B}_{S}$ is a $\sigma$-algebra of subsets
of $S$. Fix a separable Hilbert space $\mathscr{H}$, and let $\mathscr{L}\left(\mathscr{H}\right)$
denote the algebra of all bounded operators in $\mathscr{H}$. By
$\mathscr{H}$-valued (weakly) measurable functions
\[
F:\left(S,\mathscr{B}_{S}\right)\rightarrow\mathscr{H},
\]
we mean that $s\mapsto\left\langle h,F\left(s\right)\right\rangle _{\mathscr{H}}$
is measurable, for all $h\in\mathscr{H}$. In particular, both $s\mapsto\left\Vert F\left(s\right)\right\Vert _{\mathscr{H}}$
and $s\mapsto\left\langle F_{1}\left(s\right),F_{2}\left(s\right)\right\rangle _{\mathscr{H}}$
are measurable. 
\begin{rem}
\label{rem:b1}In what follows, we shall make use of weak (Pettis)
integrals, a widely known theory for which we omit details but refer
to the following references: \cite{MR4469361,MR4442762,MR3941956}.
These integrals play a crucial role in the subsequent equations, particularly
\eqref{eq:BB1} and \eqref{eq:BB2}. For additional details regarding
operator valued integrals, readers are referred to \cite{MR4576568,MR4518571,MR4514960}. 
\end{rem}

Suppose that $F$ and $G$ are measurable $\mathscr{H}$-valued functions
$\left(S,\mathscr{B}_{S}\right)\rightarrow\mathscr{H}$. Further,
let $\mu$ be a positive measure on $\left(S,\mathscr{B}_{S}\right)$.
The induced covariance operator $C=C_{F,G}$ is defined as follows:
For $\varphi\in\mathscr{H}$, 
\begin{equation}
C\left(\varphi\right)\coloneqq\int_{S}\left\langle F\left(s\right),\varphi\right\rangle _{\mathscr{H}}G\left(s\right)\mu\left(ds\right).\label{eq:BB1}
\end{equation}

\begin{prop}
The operator $C=C_{F,G}$ is of trace class from $\mathscr{H}$ to
$\mathscr{H}$ if both $F$ and $G$ satisfy 
\begin{equation}
\int\left\Vert F\left(s\right)\right\Vert _{\mathscr{H}}^{2}\mu\left(ds\right)<\infty.\label{eq:BB2}
\end{equation}
In this case, 
\[
\text{trace}_{\mathscr{H}\rightarrow\mathscr{H}}\left(C\right)=\int_{S}\left\langle F\left(s\right),G\left(s\right)\right\rangle _{\mathscr{H}}\mu\left(ds\right).
\]
\end{prop}

\begin{proof}
The operators under the integral sign in \eqref{eq:BB1} are rank-1
operators, written in Dirac's ket-bra notation as 
\[
\left|G\left(\cdot\right)\left\rangle \right\langle F\left(\cdot\right)\right|.
\]
Note that, for all $s\in S$, we have: 
\[
\text{trace}\left|G\left(s\right)\left\rangle \right\langle F\left(s\right)\right|=\left\langle F\left(s\right),G\left(s\right)\right\rangle _{\mathscr{H}}.
\]
The assumption \eqref{eq:BB2} then allows us to take the trace under
the integral sign in \eqref{eq:BB1}, and the desired result follows.
\end{proof}
\begin{defn}
\label{def:pvm}A projection-valued measure (PVM) on $\left(S,\mathscr{B}_{S}\right)$
is a function 
\[
Q\colon\mathscr{B}_{S}\rightarrow\mathscr{L}\left(\mathscr{H}\right)
\]
with the following properties:
\begin{enumerate}
\item $Q\left(\emptyset\right)=0$, $Q\left(S\right)=I$, where $I$ is
the identity operator on $\mathscr{H}$.
\item \label{enu:pvm2}$Q\left(A\right)^{2}=Q\left(A\right)^{*}=Q\left(A\right)$,
$\forall A\in\mathscr{B}_{S}$.
\item \label{enu:pvm3}$Q\left(A\cap B\right)=Q\left(A\right)Q\left(B\right)$,
$\forall A,B\in\mathscr{B}_{S}$.
\item $Q$ is $\sigma$-additive. That is, for all $\left\{ A_{i}\right\} _{i\in\mathbb{N}}$
in $\mathscr{B}_{S}$, with $A_{i}\cap A_{j}=\emptyset$ for $i\neq j$,
we have
\begin{equation}
Q\left(\cup_{i}A_{i}\right)=\sum_{i}Q\left(A_{i}\right),\label{eq:B1}
\end{equation}
where the limit holds in the strong operator topology. In particular,
\[
\left\langle h,Q\left(\cup_{i}A_{i}\right)k\right\rangle _{\mathscr{H}}=\sum_{i}\left\langle h,Q\left(A_{i}\right)k\right\rangle _{\mathscr{H}},\;\forall h,k\in\mathscr{H}.
\]
Note that $\left\langle h,Q\left(A_{i}\right)k\right\rangle _{\mathscr{H}_{i}}$
are scalar-valued measures on $\left(S,\mathscr{B}_{S}\right)$. 
\end{enumerate}
\end{defn}

\begin{defn}
A positive operator-valued measure (POVM) on a measurable space $\left(S,\mathscr{B}_{S}\right)$
is a function 
\[
Q\colon\mathscr{B}_{S}\rightarrow\mathscr{L}\left(\mathscr{H}\right)
\]
satisfying the following conditions:
\begin{enumerate}
\item $Q\left(\emptyset\right)=0$, $Q\left(S\right)=I$;
\item $Q\left(E\right)\in\mathscr{L}\left(\mathscr{H}\right)_{+}$ for every
$E\in\mathscr{B}_{S}$, i.e., $Q\left(E\right)$ is a positive (self-adjoint)
operator on $\mathscr{H}$; 
\item Equation \eqref{eq:B1} holds true.
\end{enumerate}
\end{defn}

We point the reader to the following fundamental result, \prettyref{thm:B5}.
A detailed proof is omitted, as it is widely available in the literature;
see, e.g., \cite{MR3018003,MR4502940,MR4254690,MR3991119,MR2371996,MR1435288}. 
\begin{thm}[Dilation]
\label{thm:B5}Suppose $Q\colon S\rightarrow\mathscr{L}\left(\mathscr{H}\right)$
is a POVM, then there is a ``bigger'' Hilbert space $\mathscr{K}$,
an isometry $V\colon\mathscr{H}\rightarrow\mathscr{K}$ and a projection-valued
measure $P\colon\mathscr{B}_{S}\rightarrow\mathscr{L}\left(\mathscr{K}\right)$,
such that 
\[
Q\left(A\right)=V^{*}P\left(A\right)V,\quad\forall A\in\mathscr{B}_{S}.
\]
\end{thm}

In this context, we recall the following result on positive operator-valued
kernels, and an associated version of non-commutative Radon-Nikodym
derivatives. 
\begin{thm}[\cite{JT2024canonical,jorgensen2024hilbert}]
\label{thm:c6}Let $K:S\times S\rightarrow\mathcal{L}\left(\mathscr{H}\right)$
be an operator-valued positive kernel. Then there exists a reproducing
kernel Hilbert space (RKHS) $\mathcal{K}$, and a family of mappings
$V_{s}:\mathscr{H}\rightarrow\mathcal{K}$, $s\in S$, such that 
\begin{equation}
K\left(s,t\right)=V_{s}^{*}V_{t}.\label{eq:c4}
\end{equation}

Furthermore, for p.d. kernels $K,L:S\times S\rightarrow\mathcal{L}\left(\mathscr{H}\right)$,
it holds that 
\begin{equation}
L\leq K\label{eq:c5}
\end{equation}
if and only if there exists an operator $\Gamma$ in $\mathscr{K}$,
$0\leq\Gamma\leq I_{\mathscr{K}}$, such that 
\begin{equation}
L\left(s,t\right)=V_{s}^{*}\Gamma V_{t}.\label{eq:c6}
\end{equation}
\end{thm}

\begin{rem}
\label{rem:3-7}\ 
\end{rem}

\begin{enumerate}
\item A kernel $K:S\times S\rightarrow\mathcal{L}\left(\mathscr{H}\right)$
is said to be positive definite if
\[
\sum_{i,j=1}^{n}\left\langle h_{i},K\left(s_{i},s_{j}\right)h_{j}\right\rangle _{\mathscr{H}}\geq0
\]
for all $\left(s_{i}\right)_{1}^{n}$ in $S$, $\left(h_{i}\right)_{1}^{n}$
in $\mathscr{H}$, and all $n\in\mathbb{N}$. It is straightforward
to check that if $K$ factors as in \eqref{eq:c4}, then it is positive
definite. Thus, \eqref{eq:c4} is a characterization of $\mathcal{L}\left(\mathscr{H}\right)$-valued
p.d. kernels.
\item \label{enu:3-7-2}The ordering ``$L\leq K$'' in \eqref{eq:c5}
means that 
\[
\sum_{i,j=1}^{n}\left\langle h_{i},L\left(s_{i},s_{j}\right)h_{j}\right\rangle _{\mathscr{H}}\leq\sum_{i,j=1}^{n}\left\langle h_{i},K\left(s_{i},s_{j}\right)h_{j}\right\rangle _{\mathscr{H}}
\]
for all $\left(s_{i}\right)_{1}^{n}$ in $S$, $\left(h_{i}\right)_{1}^{n}$
in $\mathscr{H}$, and all $n\in\mathbb{N}$.
\item The operator $\Gamma$ in \eqref{eq:c6} is positive self-adjoint.
It will be referred to as the Radon-Nikodym derivative of $L$ with
respect to $K$, denoted 
\begin{equation}
\Gamma=\frac{dL}{dK}.\label{eq:c7}
\end{equation}
\item \label{enu:3-7-4}In fact, the space $\mathscr{K}$ in \prettyref{thm:B5}
may be chosen to be a RKHS, and the theorem follows from the general
approach in \prettyref{thm:c6} by setting 
\[
K\left(A,B\right)\coloneqq Q\left(A\cap B\right)
\]
for $A,B\in\mathscr{B}_{S}$.
\end{enumerate}

\section{Quantum polymorphisms}\label{sec:4}

The gist of classical polymorphism in the Rokhlin setting is a disintegration.
In this section, we (1) present a non-commutative framework, and (2)
offer a non-commutative Rokhlin-type result (\prettyref{thm:4-2})
that serves as a direct integral disintegration. As in the nature
of the subject, one must make choices when extending from the commutative
case to the non-commutative case. 

Below, the term ``Quantum Polymorphisms'' is introduced as an extension
of classical polymorphisms. Classical polymorphisms are characterized
by joint probability distributions with specified marginals. In the
quantum case, this extension naturally involves positive operator-valued
measures (POVMs) acting in a Hilbert space, along with operator-valued
kernels and their Radon-Nikodym derivatives. It provides a structured
approach to understanding conditional distributions and correlations,
particularly in the context of quantum information and entanglement.
\begin{defn}[Quantum Polymorphism, $Q$-version]
\label{def:4-1}For $i=1,2$, let $\left(X_{i},\mathscr{B}_{X_{i}}\right)$
be a measurable space with Borel $\sigma$-algebra $\mathscr{B}_{X_{i}}$,
and $\pi_{i}:X\rightarrow X_{i}$ the canonical projection. Let $Q:\mathscr{B}_{X}\rightarrow\mathscr{H}$
be a POVM acting in a Hilbert space $\mathscr{H}$, where $X=X_{1}\times X_{2}$
with the product $\sigma$-algebra $\mathscr{B}_{X}$. Set 
\begin{equation}
Q_{i}=Q\circ\pi_{i}^{-1}\label{eq:b1}
\end{equation}
as a POVM on $\mathscr{B}_{X_{i}}$. We say $\left(Q,Q_{1},Q_{2}\right)$
forms a quantum polymorphism, denoted $Q_{1}\leftarrow Q\rightarrow Q_{2}$. 
\end{defn}

The setting for POVMs is a special case of operator-valued kernels,
as follows:
\begin{lem}
\label{lem:4-2}Let $K\left(A,B\right)=Q\left(A\cap B\right)$, $A,B\in\mathscr{B}_{X}$,
be as in part \eqref{enu:3-7-4} of \prettyref{rem:3-7}. Then $K:\mathscr{B}_{X}\times\mathscr{B}_{X}\rightarrow\mathscr{L}\left(H\right)$
is a positive operator-valued kernel. 
\end{lem}

\begin{proof}
Write $Q=V^{*}PV$ as in \prettyref{thm:B5}, with $V:\mathscr{H}\rightarrow\mathscr{K}$.
For all $\left(A_{i}\right)_{i=1}^{n}$ in $\mathscr{B}_{X}$, $\left(h_{i}\right)_{i=1}^{n}$
in $\mathscr{H}$, and $n\in\mathbb{N}$, it follows that
\begin{align*}
\sum_{i,j}\left\langle h_{i},K\left(A_{i},A_{j}\right)h_{j}\right\rangle _{\mathscr{H}} & =\sum_{i,j}\left\langle h_{i},Q\left(A_{i}\cap A_{j}\right)h_{j}\right\rangle _{\mathscr{H}}\\
 & =\sum_{i,j}\left\langle h_{i},V^{*}P\left(A_{i}\cap A_{j}\right)Vh_{j}\right\rangle _{\mathscr{H}}\\
 & =\sum_{i,j}\left\langle P\left(A_{i}\right)Vh_{i},P\left(A_{j}\right)Vh_{j}\right\rangle _{\mathscr{K}}\\
 & =\sum_{i}\left\Vert P\left(A_{i}\right)Vh_{i}\right\Vert _{\mathscr{K}}^{2}\geq0
\end{align*}
thus $K$ is positive definite. 
\end{proof}
\begin{thm}
\label{thm:4-2}Let $Q$, $Q_{i}$, $i=1,2$, be as specified in \prettyref{def:4-1}.
Then there exist commuting projection-valued measures $P_{1}:\mathscr{B}_{X_{1}}\rightarrow\mathscr{H}$
and $P_{2}:\mathscr{B}_{X_{2}}\rightarrow\mathscr{H}$, such that
for all $B\in\mathscr{B}_{X_{2}}$, the associated Radon-Nikodym derivative
is given by 
\begin{equation}
\frac{dQ\left(\cdot\times B\right)}{dQ_{1}}=P_{2}\left(B\right).\label{eq:d2-1}
\end{equation}
Similarly, for all $A\in\mathscr{B}_{X_{1}}$, 
\begin{equation}
\frac{dQ\left(A\times\cdot\right)}{dQ_{2}}=P_{1}\left(A\right).\label{eq:D3}
\end{equation}
\end{thm}

\begin{proof}
Let $Q=V^{*}PV$ be the dilation of $Q$ from \prettyref{thm:B5}.
Then, 
\begin{align}
Q\left(A\times B\right) & =V^{*}P\left(A\times B\right)V\nonumber \\
 & =V^{*}P\left(\left(A\times X_{2}\right)\cap\left(X_{1}\times B\right)\right)V\nonumber \\
 & =V^{*}P\left(A\times X_{2}\right)P\left(X_{1}\times B\right)V\nonumber \\
 & =V^{*}P_{1}\left(A\right)P_{2}\left(B\right)V\label{eq:D4}
\end{align}
where $P_{1}\left(A\right)=P\left(A\times X_{2}\right)$, and $P_{2}\left(B\right)=P\left(X_{1}\times B\right)$,
for all $A\in\mathscr{B}_{X_{1}}$ and $B\in\mathscr{B}_{X_{2}}$.
It is clear that $P_{i}$'s are commuting PVMs. 

Note that 
\[
V^{*}P_{1}\left(\cdot\right)V=Q_{1}\left(\cdot\right),
\]
and 
\[
V^{*}P_{1}\left(\cdot\right)P_{2}\left(B\right)V=Q\left(\cdot\times B\right)
\]
for every fixed $B\in\mathscr{B}_{X_{2}}$. 

It then follows from \prettyref{thm:c6} and \prettyref{lem:4-2}
that $P_{2}\left(B\right)$ is the Radon-Nikodym derivative of $Q\left(\cdot\times B\right)$
with respect to $Q_{1}$, which is \eqref{eq:d2-1}. The proof of
\eqref{eq:D3} is similar and is omitted. 

Specifically, we consider two operator-valued kernels defined on $\mathscr{B}_{X_{1}}\times\mathscr{B}_{X_{1}}$:
For fix $B\in\mathscr{B}_{X_{2}}$, let
\[
L\left(A_{1},A_{2}\right):=Q\left(\left(A_{1}\cap A_{2}\right)\times B\right),
\]
and 
\[
K\left(A_{1},A_{2}\right):=Q_{1}\left(A_{1}\times A_{2}\right).
\]
By \prettyref{lem:4-2}, it is easy to see that both $K$ and $L$
are positive definite. The identity \eqref{eq:D4} combined with \prettyref{thm:c6}
implies that $L\leq K$ (see part \eqref{enu:3-7-2} of \prettyref{rem:3-7}),
and the associated Radon-Nikodym derivative is given by $P_{2}\left(B\right)$. 
\end{proof}
\begin{cor}
Fix $B\in\mathscr{B}_{X_{2}}$, consider $Q\left(\cdot\times B\right)$
and $Q_{1}$ as p.d. kernels on $\mathscr{B}_{X_{1}}\times\mathscr{B}_{X_{1}}$.
Then 
\[
Q\left(\cdot\times B\right)\leq Q_{1}\left(\cdot\right).
\]
Similarly, for fixed $A\in\mathscr{B}_{X_{1}}$, 
\[
Q\left(A\times\cdot\right)\leq Q_{2}\left(\cdot\right).
\]
 
\end{cor}

\begin{cor}
Let $Q$ and $Q_{i}=Q\circ\pi_{i}^{-1}$, $i=1,2$, be as specified
in \eqref{eq:b1}. Let $\mathscr{H}_{Q}$, resp. $\mathscr{H}_{Q_{i}}$,
be the corresponding Hilbert spaces; see \prettyref{thm:4-2}. Then
$Q_{i}\rightarrow Q$ induces a contraction $f_{i}\mapsto f_{i}\circ\pi_{i}$
from $\mathscr{H}_{Q_{i}}$ into $\mathscr{H}_{Q}$, i.e., we have
\[
\left\Vert f_{i}\circ\pi_{i}\right\Vert _{\mathscr{H}_{Q}}\leq\left\Vert f_{i}\right\Vert _{\mathscr{H}_{Q_{i}}},
\]
for all $f_{i}\in\mathscr{H}_{Q_{i}}$, $i=1,2$. 
\end{cor}

\begin{thm}
\label{thm:4-5}Let $Q:\mathscr{B}_{X}\rightarrow\mathcal{L}\left(\mathscr{H}\right)$
be a POVM, $Q_{i}=Q\circ\pi_{i}^{-1}$, where $\pi_{i}:X=X_{1}\times X_{2}\rightarrow X_{i}$,
$i=1,2$, and 
\[
Q\left(A\times B\right)=V^{*}P_{1}\left(A\right)P_{2}\left(B\right)V
\]
as in \eqref{eq:D4}. Then 
\begin{equation}
Q_{1}\otimes Q_{2}:\mathscr{B}_{X}\rightarrow\mathcal{L}\left(\mathscr{H}\otimes\mathscr{H}\right)
\end{equation}
is a POVM in $\mathscr{H}\otimes\mathscr{H}$, and 
\begin{equation}
Q_{1}\otimes Q_{2}\left(A\times B\right)=\left(V\otimes V\right)^{*}P_{1}\left(A\right)\otimes P_{2}\left(B\right)\left(V\otimes V\right),\label{eq:D6}
\end{equation}
for all $A\in\mathscr{B}_{X_{1}}$, $B\in\mathscr{B}_{X_{2}}$.

Moreover, 
\begin{equation}
\frac{dQ_{1}\otimes Q_{2}\left(\cdot\times B\right)}{dQ_{1}\otimes Q_{2}\left(\cdot\times X_{2}\right)}=I\otimes P_{2}\left(B\right),\label{eq:D7}
\end{equation}
and 
\begin{equation}
\frac{dQ_{1}\otimes Q_{2}\left(A\times\cdot\right)}{dQ_{1}\otimes Q_{2}\left(X_{1}\times\cdot\right)}=P_{1}\left(A\right)\otimes I.\label{eq:D8}
\end{equation}
\end{thm}

\begin{proof}
For \eqref{eq:D6}, it suffices to check the following: Let $a\otimes b\in\mathscr{H}\otimes\mathscr{H}$,
then 
\begin{eqnarray*}
 &  & \left\langle a\otimes b,Q_{1}\left(A\right)\otimes Q_{2}\left(B\right)\left(a\otimes b\right)\right\rangle _{\mathscr{H}\otimes\mathscr{H}}\\
 & = & \left\langle a,Q_{1}\left(A\right)a\right\rangle _{\mathscr{H}}\left\langle b,Q_{2}\left(B\right)b\right\rangle _{\mathscr{H}}\\
 & = & \left\langle a,V^{*}P_{1}\left(A\right)Va\right\rangle _{\mathscr{H}}\left\langle b,V^{*}P_{2}\left(B\right)Vb\right\rangle _{\mathscr{H}}\\
 & = & \left\langle a\otimes b,\left(V\otimes V\right)^{*}P_{1}\left(A\right)\otimes P_{2}\left(B\right)\left(V\otimes V\right)\left(\left(a\otimes b\right)\right)\right\rangle _{\mathscr{H}\otimes\mathscr{H}}.
\end{eqnarray*}

The right-side of \eqref{eq:D6} is a dilation of $Q_{1}\otimes Q_{2}$,
and it can be written as 
\[
Q_{1}\otimes Q_{2}\left(A\times B\right)=\left(V\otimes V\right)^{*}\left(P_{1}\left(A\right)\otimes I\right)\left(I\otimes P_{2}\left(B\right)\right)\left(V\otimes V\right),
\]
where $P_{1}\left(A\right)\otimes I$ and $I\otimes P_{2}\left(B\right)$
are commuting PVMs in $\mathscr{H}\otimes\mathscr{H}$, defined on
$\mathscr{B}_{X_{1}}$ and $\mathscr{B}_{X_{2}}$, respectively. Therefore,
\eqref{eq:D7}--\eqref{eq:D8} follow from a similar argument as
in the proof of \prettyref{thm:4-2}.
\end{proof}

\section{Beyond the $Q$-polymorphisms}\label{sec:5}

Here, we aim to extend the $Q$-version of quantum polymorphisms (\prettyref{sec:4})
in two different settings. Recall that in the classical (commutative)
setting, a polymorphism $\mu$ for $\mu_{1}$ and $\mu_{2}$ is a
joint probability distribution having $\mu_{1},\mu_{2}$ as marginals.

\subsection{Quantum states}

One possible extension is based on the following questions:
\begin{itemize}
\item What is the right notion of polymorphism $poly(s_{1},s_{2})$ for
a fixed pair of quantum states $s_{i}$? What properties of this notion
differ from those of the classical case in our paper?
\item Can we construct quantum graphs analogous to the graphs from \cite{MR4730761},
having the sets $poly(s_{1},s_{2})$ as edges? What is the set of
extreme-points in $poly(s_{1},s_{2})$?
\item What is the role of entanglement in such an extension of the results
in our paper to quantum?
\end{itemize}
In particular, we note the following correspondence: 
\begin{align*}
\begin{matrix}\text{quantum states (trace operators)}\\
\text{on \ensuremath{\mathscr{H}_{1}\otimes\mathscr{H}_{2}}}
\end{matrix} & \Longleftrightarrow\text{joint measures/distributions}\\
\text{partial traces } & \Longleftrightarrow\text{marginal distributions}
\end{align*}

\subsection{$\mathscr{L}\left(\mathscr{H}\right)$-valued kernels}

The second approach is through operator-valued positive definite (p.d.)
kernels. In terms of kernels, an analogous notion of polymorphism
is as follows: 

Given two scalar-valued kernels $c_{1},c_{2}$ defined on $S\times S$,
a polymorphism is a $\mathcal{L}\left(\mathscr{H}\right)$-valued
p.d. kernel on $S\times S$, such that $c_{1}=tr\left(\rho_{1}K\right)$,
and $c_{2}=tr\left(\rho_{2}K\right)$. 

Here, the states $\rho_{i}$ play the role of canonical coordinate
projections. The idea is that by slicing $K$ in various ways, one
is able to recover $c_{i}$ as certain ``marginals''. 

However, without further restrictions, any two scalar valued p.d.
kernels can be linked this way:
\begin{lem}
For $c_{1},c_{2}:S\times S\rightarrow\mathbb{C}$, p.d., there exists
an $\mathcal{L}\left(\mathscr{H}\right)$-valued p.d. kernel $K$
on $S\times S$, and positive trace class operators $\rho_{1},\rho_{2}$,
such that $c_{i}=tr\left(\rho_{i}K\right)$. 
\end{lem}

\begin{proof}
Given $c_{1},c_{2}$, let $\mathscr{H}$ be any Hilbert space, and
set $\rho_{1}=\left|a\left\rangle \right\langle a\right|$, $\rho_{2}=\left|b\left\rangle \right\langle b\right|$,
where $a,b$ are unit vectors in $\mathscr{H}$ and $a\perp b$. Then
the kernel $K=c_{1}\otimes\rho_{1}+c_{2}\otimes\rho_{2}$ has the
desired properties. 
\end{proof}
Several key questions arise from this context:
\begin{itemize}
\item What are the scalar valued covariance kernels $c$ which admit such
representations?
\item Fix $K$, find $\text{SVK}(K)\coloneqq$ \{all scalar valued covariance
kernels $c$ with this representation for some density $\rho$\}.
\item Fix $\rho$, find $\text{SVK}(\rho)\coloneqq$ \{all scalar valued
covariance kernels $c$ with this representation for some $K$\}.
\item Introduce an action on all scalar valued covariance kernels $c$ via
pairs $(K,\rho)$. Is that action transitive? What about the case
when $K$ is the same for two scalar valued covariance kernels?
\end{itemize}
Some observations related to our discussion of normal states, i.e.,
those from normalized trace class operators $\rho$, and the corresponding
scalar valued p.d. kernels $trace\left(\rho K\left(\cdot,\cdot\right)\right)$
defined from some operator valued $K$: 
\begin{enumerate}
\item Compact operators are precisely the norm limits of finite rank operators,
which means that trace class operators, being compact, fit this framework.
\item There is a correspondence between $\mathcal{L}\left(\mathscr{H}\right)$-valued
p.d. kernels $K$ on $S\times S$, and scalar valued p.d. kernels
$\tilde{K}$ on $T\times T$ , where $T=\left(S,\mathscr{H}\right)$. 
\end{enumerate}
Questions:
\begin{itemize}
\item Can we use the RKHS $\mathscr{H}(\tilde{K})$ to understand the $K$
induced scalar valued states $trace\left(\rho K\left(\cdot,\cdot\right)\right)$
?
\item How does the non-product form for $\rho$ (entanglement states) reflect
itself in the properties of $trace\left(\rho K\left(\cdot,\cdot\right)\right)$
? 
\item Are polymorphisms the natural framework for addressing these questions?
\\
\end{itemize}

\bibliographystyle{amsalpha}
\bibliography{ref}

\end{document}